\documentclass{article}

\usepackage{amssymb,amsmath,amsthm}
\usepackage{geometry}

\newtheorem{theorem}{Theorem}
\newtheorem{lemma}{Lemma}
\newtheorem{claim}{Claim}
\newtheorem{remark}{Remark}
\newtheorem{corollary}{Corollary}
\newtheorem{definition}{Definition}
\newtheorem{example}{Example}
\newtheorem{conjecture}{Conjecture}
\newtheorem{observation}{Observation}
\newtheorem{question}{Question}

\def\rk{\mathrm{rk}_2}
\def\F{\mathbb{F}}
\def\Z{\mathbb{Z}}
\def\Q{\mathbb{Q}}

\title{A graph polynomial for independent sets of bipartite graphs}

\author{Qi Ge\thanks{ Department of Computer Science, University
of Rochester, Rochester, NY 14627.  Email: \{qge,stefanko\}@cs.rochester.edu.
Research supported, in part, by NSF grant CCF-0910584.} \and Daniel
\v{S}tefankovi\v{c}$^*$}

\begin{document}

\maketitle

\begin{abstract}
We introduce a new graph polynomial that encodes interesting
properties of graphs, for example, the number of matchings and
the number of perfect matchings. Most importantly, for
bipartite graphs the polynomial encodes the number of
independent sets (\#{\sc BIS}).

We analyze the complexity of exact evaluation of the polynomial at
rational points and show that for most points exact evaluation is
\#P-hard (assuming the generalized Riemann hypothesis) and
for the rest of the points exact evaluation is trivial.

We conjecture that a natural Markov chain can be used to
approximately evaluate the polynomial for a range of parameters.
The conjecture, if true, would imply an approximate counting algorithm for
\#{\sc BIS}, a problem shown, by~\cite{MR2044886}, to be complete (with
respect to, so called, AP-reductions) for a rich logically defined sub-class of \#P.
We give a mild support for our conjecture by proving that the Markov chain is
rapidly mixing on trees. As a by-product we show that the ``single bond flip''
Markov chain for the random cluster model is rapidly mixing
on constant tree-width graphs.
\end{abstract}

\section{Introduction}

Graph polynomials are a well-developed area useful for analyzing
properties of graphs (see, e.\,g., the following survey papers
~\cite{citeulike:3973700,citeulike:3973702}).
Arguably the most intriguing graph polynomial is the Tutte
polynomial~\cite{MR0018406,MR0061366}. The partition function of
the random cluster model from statistical mechanics provides
a particularly simple definition: for a graph $G=(V,E)$ let
\begin{equation}\label{eq:rc}
Z(G;q,\mu) = \sum_{S \subseteq E} q^{\kappa(S)}\mu^{|S|},
\end{equation}
where $\kappa(S)$ is the number of connected components of the
graph $(V,S)$. (It is well-known that the Tutte polynomial is
obtained from $Z$ by a simple transformation, see, e.\,g.,
equation~\eqref{eq:tutte_rc} in Section~\ref{sec:complexity}.)
The Tutte polynomial includes many graph
polynomials as special cases, e.\,g., the chromatic polynomial,
the flow polynomial, and the Potts model (see,
e.\,g.,~\cite{MR1245272}).

Now we define our graph polynomial.
\begin{definition}
The $R_2$-polynomial of a graph $G=(V,E)$ is
\begin{equation}\label{eq:poly}
R_2(G;q,\mu)=\sum_{S\subseteq E} q^{\rk(S)} \mu^{|S|},
\end{equation}
where $\rk(S)$ is the rank of the adjacency matrix of $(V,S)$
over $\F_2$ (the field with $2$ elements).
\end{definition}

Now we look at how $R_2(G;q,\mu)$ encodes some properties of graphs. For
$q = \mu^{-1/2}$ equation \eqref{eq:poly} becomes
\begin{equation*}
P(G;\mu):=R_2(G;\mu^{-1/2},\mu)=\sum_{S\subseteq E(G)}
\mu^{|S|-\rk(S)/2}.
\end{equation*}
We claim that $P(G;0)$ is the number of matchings in $G$.  To see this, note
that $\rk(S)\leq 2|S|$ (since adding an edge to $S$
changes two entries in the adjacency matrix and hence can change
rank by at most two), and $\rk(S)< 2|S|$ if $S$ is not a
matching (since rank of the adjacency matrix of a star is
$2<2|S|$, and adding further edges preserves the strict
inequality).

Let
\begin{equation*}
P(G;t,\mu):= t^{|V|} R_2(G;1/t,\mu)\quad\mbox{and} \quad P_2(G;\mu):= \mu^{-|V|/2} P(G;0,\mu).
\end{equation*}
Then $P_2(G;0)$ is the number of perfect matchings of $G$. To see this, note
that that only subsets with full rank adjacency matrix contribute to
$P(G;0,\mu)$, and then only the minimal cardinality subsets with full
rank adjacency matrix contribute to $P_2(G;0)$ (these subsets are exactly
the perfect matchings).

From now on we focus solely on bipartite graphs. For a {\em
bipartite} graph $G=(U\cup W,E)$ we let
\begin{equation}\label{eq:pwwoly}
R'_2(G;\lambda,\mu)=\sum_{S\subseteq E} \lambda^{\rk(S)} \mu^{|S|},
\end{equation}
where $\rk(S)$ is the rank of the {\em bipartite} adjacency matrix of
$(U\cup W,S)$. Note that
\begin{equation}\label{errp}
R_2(G;\lambda,\mu)=R'_2(G;\lambda^2,\mu),
\end{equation} since the
adjacency matrix contains ``two copies'' of the bipartite
adjacency matrix (one of them transposed). (The reason for
definition \eqref{eq:pwwoly} is that we prefer to operate with
bipartite adjacency matrix for bipartite graphs.)

In Section~\ref{S:BIS} we prove that $R'_2$ counts the number of independent sets in bipartite graphs.
\begin{theorem}\label{t:1}
Let $G=(U\cup W,E)$ be a bipartite graph. The number of independent sets of $G$ is
given by
$$
2^{|U|+|W|-|E|}R'_2(G;1/2,1).
$$
\end{theorem}

We show in Section~\ref{sec:complexity} that exact evaluation of
the polynomial $R'_2(G;\lambda,\mu)$ is \#P-hard at a variety of
rational points $(\lambda,\mu)$ assuming the validity of the
generalized Riemann hypothesis (GRH).
\begin{theorem}\label{thm:comlexity}
Exact evaluation of $R'_2$ at rational point $(\lambda,\mu)$ is
\begin{itemize}
\item polynomial-time computable when $\lambda \in \{0,1\}$ or
$\mu=0$ or $(\lambda,\mu)=(1/2,-1)$;
\item \#P-hard when $\lambda
\not\in \{0,1,1/2\}$ and $\mu \neq 0$, assuming GRH;
\item
\#P-hard when $\lambda=1/2$ and $\mu \not\in \{0,-1\}$.
\end{itemize}
\end{theorem}

\begin{remark}
For the non-bipartite case we have the following classification.
Exact evaluation of $R_2$ at rational point $(\lambda,\mu)$ is
polynomial-time computable when $\mu=0$ or $\lambda\in\{-1,0,1\}$;
the $\lambda=-1$ case follows from the fact that a skew-symmetric
matrix with zero diagonal has even rank over any field (the
zero diagonal condition is redundant for fields of
characteristic $\neq 2$). For any other rational $\lambda$ and $\mu$
we get \#P-hardness of evaluating the $R_2$ polynomial from
Theorem~\ref{thm:comlexity} and~\eqref{errp} (again assuming GRH).
(Note that $(\lambda,\mu)\mapsto(\lambda^2,\mu)$ never maps
to the easy case $(1/2,-1)$, since $\lambda$ is rational.
It would be nice to have hardness classification of evaluating $R_2$
and $R_2'$ for, say, algebraic $\lambda$ and $\mu$.)
\end{remark}

Because of the hardness of exact evaluation of $R'_2$, we turn to
approximate evaluation of $R'_2(G;\lambda,\mu)$.

\begin{remark}
There is a fully polynomial randomized approximation scheme (FPRAS) for $R'_2(G;\lambda,\mu)$ when
$\lambda=1/2$ and $-1 < \mu < 0$. This follows from Theorem~\ref{thm:pbis} (stated in Section~\ref{sec:complexity})
and the fact that there is an FPRAS for \#{\sc PBIS}$(\eta)$ (defined in Section~\ref{sec:complexity})
when $0 < \eta < 1$ (see~\cite{MR1995687}, and note that \#{\sc PBIS}$(\eta)$ corresponds to $\beta=1$, $\gamma=(1+\eta)/(1-\eta)$ and $\mu=1$ in their parametrization).
\end{remark}

We now define the sampling problem associated with $R'_2$.

\vskip 0.2cm
\noindent\textsc{Rank Weighed Subgraphs} with $\lambda,\mu\geq 0$, ({\sc RWS}$(\lambda,\mu)$)

Instance: a bipartite graph $G=(U\cup W,E)$,

Output: $S\subseteq E$ with probability of $S\propto \lambda^{\rk(S)}\mu^{|S|}$.

\vskip 0.2cm

The ``single bond flip'' chain is a natural approach to sampling
from {\sc RWS}$(\lambda,\mu)$.
\begin{definition}\label{d1}
{\em Single bond flip} chain is defined as follows: pick an edge
$e\in E$ at random and let $S=X_t\oplus \{e\}$. Set $X_{t+1}=S$
with probability
$(1/2)\min\{1,\lambda^{\rk(S)-\rk(X_t)}\mu^{|S|-|X_t|}\}$
and $X_{t+1}=X_t$ with the remaining probability.
\end{definition}

In each step of the single bond flip chain, we have to compute the
rank of a matrix over $\F_2$ (corresponding to $S$) which differs
from the current matrix (corresponding to $X_t$) in a single
entry. One can use dynamic matrix rank problem algorithms
to perform this computation in $O(n^{1.575})$ arithmetic
operations per step~\cite{MR2305543}.

Instead of flipping one edge in a step, we can have another Markov
chain which flips a random subset of edges adjacent to a single
vertex. It seems likely that the new chain can generate good random samples
faster than the single bond flip chain---a step of the new chain can be performed in $O(n^2)$
arithmetic operations (using ``rank one update'' for the
dynamic matrix rank problem~\cite{MR2305543}).

We optimistically conjecture that the single bond flip chain (and
hence the other chain) mixes rapidly.
\begin{conjecture}\label{c1}
For any fixed $\lambda,\mu>0$ the single bond flip chain mixes in
polynomial time.
\end{conjecture}
Note that Goldberg and Jerrum~\cite{GJ10} conjecture the opposite,
in particular for $\lambda=1/2$ and $\mu=1$. We make the
conjecture based on similarity of $R_2'$ polynomial to the random
cluster model and the fact that slow mixing of the single bond
flip chain was not established in random cluster model (slow
mixing is usually easier to establish than rapid mixing).

Conjecture~\ref{c1} could be false for small $\lambda$ and large
$\mu$ if the problem of maximizing $|S|-c\cdot{\rm rk}_2(S)$ ($c$ is a constant and $c>1$) over
$S\subseteq E$ is hard (the corresponding maximization problem is
easy for the random cluster model).

We can provide a mild support for Conjecture~\ref{c1}. In Section~\ref{sec:mix_mm} we prove that for fixed $\lambda,\mu > 0$ the single bond flip
chain mixes, in time polynomial in the number of vertices, for trees.
\begin{theorem}\label{thm:mix_mm}
For every fixed $\lambda,\mu > 0$, the mixing time $\tau(\varepsilon)$ of the single bond flip chain for a tree on $n$ vertices is
\begin{equation*}
\tau(\varepsilon)=O\left(n^{3+|\log_2\lambda|}(|\log\lambda|+|\log\mu|+\log(1/\varepsilon))\right).
\end{equation*}
\end{theorem}

As a by-product of our techniques, we show that single bond flip Markov chain for the
random cluster model is rapidly mixing if $q,\mu > 0$ and $G$ has constant tree-width
(the condition $q,\mu>0$ is equivalent to $x,y>1$ for the Tutte polynomial
$T(G;x,y)$).

\section{Independent sets in bipartite graphs}\label{S:BIS}

The problem of counting independent sets (\#{\sc IS}) in a
graph is of interest in both computer science and statistical
physics (independent sets are a special case of the so-called
hard-core model, see, e.\,g.,~\cite{MR998375}). Exact computation
of \#{\sc IS} is \#P-complete even for $3$-regular planar bipartite
graphs~\cite{MR1861282,MR2354227}. Fully polynomial randomized
approximation scheme (FPRAS) is known for graphs with maximum
degree $\Delta\leq 5$, \cite{MR1716763,MR1747721,MR2277139}. Unless
RP$=$NP, an FPRAS does not exist for graphs with $\Delta\geq
25$,~\cite{MR1936657} .

Now we focus on the problem of counting independent sets in
bipartite graphs (\#{\sc BIS}). While for exact counting the
complexity of \#{\sc BIS} and \#{\sc IS} is the same, the
situation looks very different for approximate counting, for
example, no inapproximability result is known for \#{\sc BIS}.
Dyer et al.~\cite{MR2044886} show that \#{\sc BIS} is complete
w.r.t. approximation-preserving reductions (AP-reductions) in a
sub-class of \#P. Many problems were shown to be equivalent
(w.r.t. AP-reductions) to \#{\sc BIS}, for example, \#{\sc
Downsets}, \#{\sc 1p1nSat}~\cite{MR2044886}, computing the
partition function of a ferromagnetic Ising model with local
fields~\cite{MR2286511}, and counting the number of satisfying
assignments of a class of Boolean CSP instances~\cite{Dyer2009}. A
pertinent negative result for \#{\sc BIS} is that Glauber dynamics
(or more generally, any chain whose states are independent sets
and that flips at most $0.35n$ vertices in one step) cannot be
used to efficiently sample random independent sets in a random
$6$-regular bipartite graphs on $n+n$ vertices~\cite{MR1936657}.

The rest of this section is devoted to proving Theorem~\ref{t:1}.
It will be convenient to work with matrices instead of graphs. For
two zero-one matrices $A,B$ we say $B\leq A$ if $B$ corresponds to
a subgraph of $A$, formally

\begin{definition}
Let $A,B$ be zero-one $n_1 \times n_2$ matrices. We say $B\leq A$ if $A_{ij}=0$ implies $B_{ij}=0$, for all $i \in [n_1]$ and $j\in [n_2]$.
Let ${\cal C}_A$ be the set of zero-one $n_1 \times n_2$ matrices $B$ such that $B\leq A$.
\end{definition}

Let $\#_1(A)$ denote the number of ones in $A$ (that is, the
number of edges in the corresponding graph). The {\sc RWS} problem
rephrased for matrices is:

\vskip 0.2cm
\noindent\textsc{Rank Weighed Matrices} with $\lambda,\mu \geq 0$ ({\sc RWM}$(\lambda,\mu)$)

Instance: an $n_1\times n_2$ matrix $A$.

Output: $B\in {\cal C}_A$ with probability of $B\propto \lambda^{\rk(B)}\mu^{\#_1(B)}$.
\vskip 0.2cm

The problem of sampling independent sets in bipartite graphs is:

\vskip 0.2cm
\noindent\textsc{Bipartite Independent Sets} ({\sc BIS})

Instance: a bipartite graph $G=(U\cup W,E)$.

Output: a uniformly random independent set of $G$.
\vskip 0.2cm

Before we show a connection between {\sc BIS} and {\sc
RWM}$(1/2,1)$ we remark that to sample bipartite
independent sets it is enough to sample a subset of one side, say
$U$, from the correct (marginal) distribution. We now describe
this distribution in a setting which will be advantageous for
the proof of Theorem~\ref{t:1}.

We will represent an independent set by a pair of (indicator) vectors $u,v$ (where
$u\in\F_2^{n_1}$ and $v\in\F_2^{n_2}$).

\begin{definition}
We say that two vectors $\alpha,\beta\in\F_2^{n}$ {\em share a one} if there exists $i\in [n]$ such
that $\alpha_i = \beta_i = 1$.
\end{definition}

We will use the following simple fact.

\begin{observation}
Let $\alpha,\beta\in\F_2^{n}$. Let $d$ be the number of ones in $\beta$.
\begin{itemize}
\item
if $\alpha,\beta$ share a one then there are $2^{d-1}$ vectors $\beta'\leq \beta$ such that $\alpha^{\rm T} \beta'\equiv 0 \mod 2$.
\item
if $\alpha,\beta$ do not share a one then there are $2^{d}$ vectors $\beta'\leq \beta$ such that $\alpha^{\rm T} \beta'\equiv 0 \mod 2$.
\end{itemize}
\end{observation}

Let $u\in\F_2^{n_1}$ be a vector. We would like to count the number of $v\in\F_2^{n_2}$
such that $u,v$ is an independent set. Note that $u,v$ is an independent set
iff $v_j=0$ for every $j\in [n_2]$ such that $u$ and $j$-th column of $A$ share a one.
Let $k$ be the number of columns of $A$ that do not share a one with~$u$.
Then we have
\begin{equation}\label{e1}
\mbox{$u\in\F_2^{n_1}$ occurs in $2^k$ independent sets.}
\end{equation}
Thus to sample independent sets in a bipartite graph $G$ with $n_1 \times n_2$ bipartite adjacency matrix $A$
it is enough to sample $u\in\F_2^{n_1}$ with the probability of $u$ proportional to $2^k$,
where $k$ is the number of columns of $A$ that do not share a one with $u$. We will call
this distribution on $u$ the marginal {\sc BIS} distribution.

The following lemma shows a tight connection between {\sc BIS} and {\sc RWM}$(1/2,1)$---given a sample from
one distribution it is trivial to obtain a sample from the other one.

\begin{lemma}\label{l1}
Let $G$ be a bipartite graph with bipartite adjacency matrix $A$.
\begin{itemize}
\item Let $u,v$ be a uniformly random independent set of $G$. Let $B$ be a uniformly random matrix from
the following set $\{ D\in{\cal C}_{A} \,|\,u^{\rm T} D \equiv 0 \mod 2\}$. Then $B$ is from the {\sc RWM}$(1/2,1)$-distribution.
\item Let $B\in {\cal C}_{A}$ be a random matrix from the {\sc RWM}$(1/2,1)$-distribution. Let
$u\in\F_2^{n_1}$ be a uniformly random vector from the left null space of $B$
(that is, $\{\beta\in\F_2^{n_1}\,|\,\beta^{\rm T} B \equiv 0\mod 2\}$). Then $u$ is from
the marginal {\sc BIS} distribution.
\end{itemize}
\end{lemma}

\begin{proof}
Let $Q$ be the set of $u,B$ pairs such that $u^{\rm T} B\equiv 0\mod 2$ and $B\leq A$.
Let $\psi$ be the uniform distribution on $Q$. Note that $\psi$ marginalized
over $u$ yields the {\sc RWM}$(1/2,1)$-distribution on $B\leq
A$, here we are using the fact that a $d$-dimensional space (in
this case the left null space of $B$) over $\F_2$ has $2^d$
elements. Formally,
\begin{equation}\label{ee11}
P(B)=\sum_{u:u^{\rm T} B \equiv 0 \mod{2}} \frac{1}{|Q|} = \frac{2^{n_1-\rk(B)}}{|Q|}
= \frac{2^{-\rk(B)}}{R_2'(G;1/2,1)}.
\end{equation}

Next we show that $\psi$ marginalized over $B$ yields the marginal
{\sc BIS} distribution. We compute the number of $B\leq A$ such
that $u^{\rm T} B\equiv 0\mod 2$. Let us use the same $k$ as
in~\eqref{e1}, that is, $k$ is the number of columns of $A$ that
do not share a one with $u$.

Note that the columns of $B$ can be chosen independently and only if the column and $u$ share a
one is the number of choices (for that column) halved. Let $\#_1(A)$ be the number of ones in $A$. Thus
\begin{equation}\label{e2}
\mbox{there are $2^{\#_1(A)-(n_2-k)}$ choices of $B\leq A$ such that $u^{\rm T} B\equiv 0\mod 2$.}
\end{equation}
Note that for fixed $u$ the counts in \eqref{e1} and \eqref{e2}
differ by a factor of $2^{\#_1(A)-n_2}$ (which is independent of
$u$). Thus $\psi$ marginalized over $B$ yields the marginal {\sc
BIS} distribution on $u$. Formally
\begin{equation}\label{ee22}
P(u)=\frac{2^{\#_1(A)-(n_2-k)}}{|Q|}=\frac{2^k}{\#\mathrm{BIS}(G)}.
\end{equation}

Note that this proves both claims of the lemma since in both cases the $u,B$ pair is from $\psi$
(by first sampling from a marginal and then sampling the remaining variable) and the conclusion
in both claims is a statement about marginal (of the remaining variable).
\end{proof}

Theorem~\ref{t:1} now follows from the proof of Lemma~\ref{l1}.

\begin{proof}[Proof of Theorem~\ref{t:1}]
Let $Q$ be the set from the proof of Lemma~\ref{l1}. From~\eqref{ee11} we obtain
\begin{equation}\label{ee33}
|Q|=R'_2(G;1/2,1) 2^{n_1}.
\end{equation}
From~\eqref{ee22} we have that the number of independent sets of $G$ is given by
\begin{equation}\label{ee44}
\#\mathrm{BIS}(G)=\frac{|Q|}{2^{\#_1(A)-n_2}}.
\end{equation}
Combining \eqref{ee33} and \eqref{ee44} we obtain the theorem.
\end{proof}

We do not know a good combinatorial interpretation for the mod-$2$ rank of $B$ for general graphs.
For forests (which are, of course, always bipartite) we have the following characterization.

\begin{lemma}\label{l2}
Let $G=(V,E)=(U\cup W,E)$ be a forest with bipartite adjacency matrix $A$. Then $\rk(A)$ is the size
of maximum matching in $G$.
\end{lemma}

\begin{proof}
Let $a\in V$ be a leaf of $G$ and let $e=\{a,b\}\in E$ be the edge adjacent to $a$.
Note that $b$ is matched in every maximum matching $M$ (otherwise one could add $e$ to $M$).
Thus removing $b$ and all adjacent edges decreases the size of maximum matching by $1$.

Now we argue that removing $b$ (and all adjacent edges) also decreases rank (over $\F_2$) by $1$.
W.l.o.g. assume that $b$ corresponds to the first row and $a$ corresponds to the first column.
Removing $b$ (and all adjacent edges) corresponds to removing the first row of $A$. Note
that this decreases rank by at most $1$ and it does decrease it by $1$, since the
only non-zero entry in the first column is in the first row.
\end{proof}

\section{Exact evaluation of $R'_2$ (proof of Theorem~\ref{thm:comlexity})}\label{sec:complexity}

We will prove Theorem~\ref{thm:comlexity} in this section. Let $G=(V,E)=(U\cup W,E)$ be
a bipartite graph. First we deal with the cases
where exact evaluation of $R'_2(G;\lambda,\mu)$ is easy. For cases $\lambda\in\{0,1\}$ and
$\mu=0$ we have
$$
R_2'(G;0,\mu)=R_2'(G;\lambda,0)=1\quad\quad\mbox{and}\quad\quad R_2'(G;1,\mu)=(1+\mu)^{|E|}.
$$
For $\lambda=1/2$ and $\mu=-1$ we will show
\begin{equation}\label{ezzz}
R_2'(G;1/2,-1)=2^{|E|-|V|+t},
\end{equation}
where $t$ is the number of isolated vertices in $G$ (see the remark after
Theorem~\ref{thm:pbis} in this section).

For the hardness results we give reductions from the problem of evaluating the Tutte polynomial
(to establish the second part of Theorem~\ref{thm:comlexity}) and \#{\sc BIS}
(to establish the third part of Theorem~\ref{thm:comlexity}).

The Tutte polynomial of a graph $G=(V,E)$ is a polynomial in two variables $x$, $y$ defined by
\begin{equation*}
T(G;x,y) = \sum_{S \subseteq E} (x-1)^{\kappa(S)-\kappa(E)}(y-1)^{|S|-|V|+\kappa(S)},
\end{equation*}
where $\kappa(S)$ is the number of connected components of the graph $(V,S)$. The Tutte polynomial
is closely related to the random cluster model (see,
e.\,g.,~\cite{MR1245272}). Let $Z(G;q,\mu)$ be defined as in \eqref{eq:rc}. One has
\begin{equation}\label{eq:tutte_rc}
T(G;x,y) = (x-1)^{-\kappa(E)}(y-1)^{-|V|}Z(G;(x-1)(y-1),(y-1)),
\end{equation}
where we assume $x\neq 1$ and $y\neq 1$.

We are going to use the following result on the complexity of exact evaluation of the Tutte polynomial.
\begin{theorem}[\cite{MR1049758}]\label{ttty}
Exact evaluation of the Tutte polynomial is \#P-hard for all rational numbers $x,y$ except
when
\begin{enumerate}
\item $(x-1)(y-1)=1$; or
\item $(x,y)$ equals $(1,1)$, $(-1,-1)$, $(0,-1)$ or $(-1,0)$.
\end{enumerate}
\end{theorem}

The second part of~Theorem~\ref{thm:comlexity} will be proved by
reducing from exact evaluation of the Tutte polynomial. We prove
the following Lemma in  Section~\ref{sec:lemreduc}.

\begin{lemma}\label{lem:reduction}
Assuming the validity of the GRH, exact
evaluation of the Tutte polynomial at $x,y$ is polynomial-time
Turing reducible to exact evaluation of $R_2'$ at $\lambda,\mu$,
when
\begin{equation}\label{j20}
(x-1)(y-1)=1/\lambda-1,\quad y-1=\mu^2,\quad \lambda \not\in \{0,1\},\quad\mbox{and}\quad \mu \neq 0.
\end{equation}
\end{lemma}

 Assuming GRH, by Lemma~\ref{lem:reduction}
and Theorem~\ref{ttty}, we have that exact evaluation of $R_2'$ at
rational point $(\lambda,\mu)$ is \#P-hard when $\lambda \not\in
\{0,1/2,1\}$ and $\mu \neq 0$. We do not get \#P-hardness for
$\lambda=1/2$ since the reduction is from the Tutte polynomial at
$(x-1)(y-1)=1$ which is polynomial-time computable (part 1 of
Theorem~\ref{ttty}). (The other easy cases of the Tutte polynomial
have no impact since $y=1$ implies $\mu=0$ and $y\in\{0,-1\}$
implies that $\mu$ is not real.) We proved the second part of
Theorem~\ref{thm:comlexity}.

Now we prove the third part of~Theorem~\ref{thm:comlexity} (the proof of
main lemmas is deferred to later sections). To show \#P-hardness of exact evaluation of $R'_2(G;1/2,\mu)$ for
$\mu\notin\{-1,0\}$, we prove a connection between $R'_2$ and the
``permissive version of \#{\sc BIS}'' (\#{\sc PBIS}) introduced
in~\cite{MR2286511}; \#{\sc PBIS}  is a generalization of \#{\sc BIS} where
the weight of a set of vertices is determined by the number of pairs of
neighboring vertices that are both in the set (in \#{\sc BIS} the weight
is zero raised to the number of such pairs).

\vskip 0.2cm
\noindent\textsc{\#Permissive Bipartite Independent Sets} with $\eta$ (\#{\sc PBIS}$(\eta)$)

Instance: a bipartite graph $G=(U\cup W,E)$.

Output: the quantity
$$\#\mathrm{PBIS}(G;\eta)=\sum_{\sigma:U \cup W
\rightarrow \{0,1\}}
(1+\eta)^{w(\sigma)}(1-\eta)^{|E|-w(\sigma)},$$
where
$w(\sigma)$ is the number of edges in $E$ with both endpoints labelled $1$ by $\sigma$.
\vskip 0.2cm

(We are using a different parametrization than~\cite{MR2286511}---our $\eta$ and their $\gamma$
are connected by $\gamma^2=(1+\eta)/(1-\eta)$.)

Note that
$$\#\mathrm{BIS}(G) = 2^{|E|} \#\mathrm{PBIS}(G;-1).$$
The following result is a generalization of
Theorem~\ref{t:1} and shows that $R'_2$ encodes \#{\sc
PBIS}$(\eta)$ as well. The proof is deferred to
Section~\ref{sec:thmpbis}.

\begin{theorem}\label{thm:pbis}
Let $G=(V,E)=(U\cup W,E)$ be a bipartite graph.
\begin{equation*}
\#\mathrm{PBIS}(G;\eta)=2^{|V|}R'_2(G;1/2,-\eta).
\end{equation*}
\end{theorem}

Note that $\#\mathrm{PBIS}(G;1)=2^{|E|+t}$, where $t$ is the number of
isolated vertices of $G$ (since the other vertices have to be labeled $1$
by $\sigma$). This implies \eqref{ezzz}.

The following result on
\#{\sc PBIS}$(\eta)$ will be proved in Section~\ref{sec:lemreduc2}.

\begin{lemma}\label{lem:reduction2}
\#{\sc BIS} is polynomial-time Turing reducible to \#{\sc PBIS}$(\eta)$
with $\eta$ a rational number and $\eta \not\in \{\pm 1,0\}$.
\end{lemma}

The third part of Theorem~\ref{thm:comlexity} follows from
Theorem~\ref{thm:pbis}, Lemma~\ref{lem:reduction2} and the fact
that exact computation of \#{\sc BIS} is \#P-complete~\cite{MR721012}.

\subsection{Reducing Tutte polynomial to $R'_2$ polynomial (proof of Lemma~\ref{lem:reduction})}\label{sec:lemreduc}

We will focus on bipartite graphs $G=(U\cup W,E)$ such that vertices in
partition $W$ have degree at most $2$ (a natural operation that
produces such graphs is {\em $2$-stretch}, that is, replacement of each
edge with a path of length $2$).

Let $G=(U\cup W,E)$ be a bipartite graph with max-degree in $W$ bounded by $2$.
We call a connected component $C=(U_C\cup W_C,E_C)$ of $G$ {\em pure} if every
vertex in $W_C$ has degree $2$ in $C$. A component that is not pure
will be called {\em mixed}. The evaluation of $R_2'$ polynomial in $G$
can be expressed using pure connected components as follows.

\begin{lemma}\label{lem:zr}
For every bipartite graph $G=(U \cup W,E)$ such that the degree of each vertex in $W$ is bounded by $2$,
\begin{equation*}
R'_2(G;\lambda,\mu)=\sum_{S \subseteq E} \lambda^{|U|-\kappa'(S)}\mu^{|S|},
\end{equation*}
where $\kappa'(S)$ is the number of pure connected components in $(U \cup W,S)$.
\end{lemma}

Before proving Lemma~\ref{lem:zr} we need the following characterization of the rank of bipartite adjacency matrices
over $\F_2$.

\begin{lemma}\label{lbn}
Let $G=(U\cup W,E)$ be a connected bipartite graph with max-degree in $W$ bounded by $2$. Let
$B$ be the adjacency matrix of $G$. Then
$$
{\rm rk}_2(B)=\Big\{\begin{array}{ll}
|U| & \mbox{if there is a vertex of degree $1$ in $W$},\\
|U|-1 & \mbox{otherwise}.
\end{array}
$$
\end{lemma}

\begin{proof}
Let $x\in\F_2^{U}$ be a solution of the linear system $x^{\rm T} B = 0$. Let $U_i$ be the set of
vertices $u\in U$ such that $x_u=i$, for $i=0,1$. Note that no vertex $v\in W$ has neighbors
both in $U_0$ and $U_1$ (otherwise $(x^{\rm T} B)_v = 1$). Thus for $G$ to be connected either
$x = 0$ or $x = 1$. If there is a vertex of degree $1$ in $W$ then $x=1$ is not a solution
and hence ${\rm rk}_2(B)=|U|$. On the other hand if all vertices in $W$ have degree $2$
then $x=1$ is a solution and hence ${\rm rk}_2(B)=|U|-1$.
\end{proof}

Now we prove Lemma~\ref{lem:zr}.

\begin{proof}[Proof of Lemma~\ref{lem:zr}]
We will show that
\begin{equation}\label{yyh}
\rk(S) = |U|-\kappa'(S),
\end{equation}
where $\rk(S)$ is the rank (over $\F_2$) of $B$, the bipartite adjacency matrix of $(U\cup W,S)$,
and $\kappa'(S)$ is the number of pure connected components of $(U\cup W,S)$.

Note that $B$ has a block structure with a block for each connected component. The rank
is the sum of the ranks of the blocks. Equation~\eqref{yyh} now follows from Lemma~\ref{lbn}.
\end{proof}

We now lay groundwork for the proof of Lemma~\ref{lem:reduction}.
We use the following construction in the reduction. Given a graph
$H=(V_H,E_H)$ and a bipartite graph $\Upsilon=(U_\Upsilon\cup W_\Upsilon,E_\Upsilon)$
with a specific vertex $u \in U_\Upsilon$, we construct a bipartite graph $G$ from $H$ and
$\Upsilon$ as follows. Let $\hat{H}=(U_{\hat{H}}\cup W_{\hat{H}},E_{\hat{H}})$ be the
$2$-stretch of $H$, where $U_{\hat{H}}$ corresponds to the vertices of $H$.
For each vertex $v$ in $U_{\hat{H}}$ we identify $v$ with $u$ in a copy of $\Upsilon$ (thus we have
$|V_H|=|U_{\hat{H}}|$ copies of $\Upsilon$). We call
the graph $G$ the {\em stretch-sum} of $H$ and $(\Upsilon,u)$. Note that
if $W_\Upsilon$ contains only vertices of degree at most $2$ then the partition of $G$
containing $W_{\hat{H}}$ contains only vertices of degree at most $2$.

We define two functions related to $R_2'$.
\begin{definition}
Let $\lambda,\mu\in {\mathbb R}$. Let $\Upsilon = (U \cup W,E)$ be a bipartite graph with a specific vertex
$u \in U$. Assume that the max-degree in $W$ bounded by $2$. We define
\begin{equation}\label{eqZP}
Z'_p(\Upsilon;\lambda,\mu) = \sum_{S} \lambda^{-\kappa'(S)}\mu^{|S|},
\end{equation}
where the sum is over all $S\subseteq E$ such that $u$
is in a pure connected component of $(U \cup W,S)$, and $\kappa'(S)$
is the number of pure connected components of $(U \cup W,S)$.

Similarly, we define
\begin{equation}\label{eqZM}
Z'_m(\Upsilon;\lambda,\mu) = \sum_{S} \lambda^{-\kappa'(S)}\mu^{|S|},
\end{equation}
where the sum is over all sets $S\subseteq E$ such that $u$
is in a mixed connected component of $(U \cup W,S)$.
\end{definition}

The following lemma provides a connection between the random cluster
partition function $Z$ of $H$ and the $R_2'$ polynomial of $G$ for
rational $\lambda$ and $\mu$.

\begin{lemma}\label{lem:zz'}
Fix rational $\lambda \not\in \{0,1\}$ and rational $\mu \neq 0$.
Let $p$ be a prime such that $\lambda \in \Z^*_p$. Let $\Upsilon=(U\cup W,E)$
be a bipartite graph with a specific vertex $u \in U$, such that the max-degree
in $W$ bounded by $2$. Suppose $\Upsilon$ satisfies
\begin{equation}\label{eq:lemmod}
\begin{split}
X:=\lambda Z'_p(\Upsilon;\lambda,\mu) \not\equiv 0 \mod p,\\
Y:=Z'_m(\Upsilon;\lambda,\mu)+\lambda Z'_p(\Upsilon;\lambda,\mu) \equiv 0 \mod p.
\end{split}
\end{equation}
Let $G$ be the stretch-sum of $H=(V_H,E_H)$ and $(\Upsilon,u)$. Then
\begin{equation}\label{j10}
R'_2(G;\lambda,\mu) \equiv \lambda^{|V_H|\cdot|U|}X^{|V_H|}Z(H;1/\lambda-1,\mu^2) \mod p,
\end{equation}
where $Z(H;1/\lambda-1,\mu^2)$ is defined in~\eqref{eq:rc}.
\end{lemma}

\begin{proof}
Let $\hat{H}=(V_{\hat{H}},E_{\hat{H}})$ be the $2$-stretch of $H$.
Note that $\hat{H}$ is a subgraph of $G=(V_G,E_G)$. Let $S_0
\subseteq E_{\hat{H}}$. Let $\Lambda_{S_0}$ be a family of subsets
$S \subseteq E_G$ such that $S \cap E_{\hat{H}}=S_0$. Now, we
evaluate
\begin{equation}\label{eq:lemzz'1}
\sum_{S \in \Lambda_{S_0}} \lambda^{-\kappa'(S)}\mu^{|S|}
\end{equation}
modulo $p$.

\begin{claim}\label{clm:lemmod}
$$\eqref{eq:lemzz'1} \equiv
\begin{cases}
(\lambda^{-1}-1)^{\kappa(S_0)}\mu^{|S_0|}X^{|V_H|} \mod p & \mbox{if $(V_{\hat{H}},S_0)$ has
no mixed connected component,}\\
0 \mod p & \mbox{otherwise,}
\end{cases}
$$
where $\kappa(S_0)$ is the number of connected components of $(V_{\hat{H}},S_0)$.
\end{claim}

\begin{proof}[Proof of Claim~\ref{clm:lemmod}]
Equation~\eqref{eq:lemzz'1} can be rewritten as a product,
where each term in the product corresponds to a connected
component of $(V_{\hat{H}},S_0)$ (since each connected component
with the copies of $\Upsilon$ attached to it influences
$\kappa'(S)$ independently). Thus
\begin{equation*}
\eqref{eq:lemzz'1} = \prod_{C}\Phi_{C},
\end{equation*}
where for each connected component $C=(V_C,E_C)$ in $(V_{\hat{H}},S_0)$ such
that there are $k$ copies of $\Upsilon$ (we refer to the copies
$\Upsilon_1,\dots,\Upsilon_k$ and to their special vertices $u_1,\dots,u_k$)
attached to it,
\begin{equation}\label{j2}
\Phi_{C} = \sum_{S_1 \subseteq E_{\Upsilon_1}} \cdots \sum_{S_k \subseteq E_{\Upsilon_k}}
\lambda^{-\kappa'(\bigcup_{i \in [k]}S_i \cup E_C)}\mu^{|E_C|+\sum_{i \in [k]}|S_i|}.
\end{equation}
Let $A_{i,0}$ be the set of $S_i$ such that $u_i$ is in a mixed component of $(V_{\Upsilon_i},S_i)$ and
let $A_{i,1}=2^{E_{\Upsilon_i}}\setminus A_{i,0}$. Equation~\eqref{j2} can be written as follows
\begin{equation}\label{j6}
\Phi_C=\sum_{x_1=0}^1\cdots\sum_{x_k=0}^1 \sum_{S_1 \in A_{1,x_1}} \cdots \sum_{S_k \in A_{k,x_k}}
\lambda^{-\kappa'(\bigcup_{i \in [k]}S_i \cup E_C)}\mu^{|E_C|+\sum_{i \in [k]}|S_i|}.
\end{equation}
We have
\begin{equation*}
\kappa'\left(\bigcup_{i \in [k]}S_i \cup E_C\right)=\sum_{i \in [k]} \kappa_i'(S_i) - \sum_{i=1}^k x_i + \ell,
\end{equation*}
where $\ell=1$ if $x_1=\cdots=x_k=1$ and $C=(V_C,E_C)$ is a pure connected component of $(V_{\hat{H}},S_0)$,
and $\ell=0$
otherwise. Thus
\begin{equation}\label{j7}
\begin{split}
\sum_{S_1 \in A_{1,x_1}} \cdots \sum_{S_k \in A_{k,x_k}}
\lambda^{-\kappa'(\bigcup_{i \in [k]}S_i \cup E_C)}\mu^{|E_C|+\sum_{i \in [k]}|S_i|}=\\
\lambda^{-\ell}\mu^{|E_C|} \sum_{S_1 \in A_{1,x_1}} \cdots \sum_{S_k \in A_{k,x_k}}
\prod_{i=1}^k \lambda^{-\kappa'_i(S_i)+x_i}\mu^{|S_i|} = \lambda^{-\ell}\mu^{|E_C|} X^{k'} (Y-X)^{k-k'},
\end{split}
\end{equation}
where $k'=x_1+\dots+x_k$.

Plugging~\eqref{j7} into~\eqref{j6} we obtain
\begin{equation}\label{j8}
\Phi_C = \mu^{|E_C|} Y^k + L,
\end{equation}
where $L=(1/\lambda-1)\mu^{|E_C|} X^k$ if $C$ is a pure component of $(V_{\hat{H}},S_0)$ and $0$ otherwise.
Evaluating~\eqref{j8} modulo $p$ (using~\eqref{eq:lemmod}) we obtain
$$
\Phi_C \equiv \Big\{\begin{array}{ll}
0 \mod p& \mbox{if $C$ is a mixed component of $(V_{\hat{H}},S_0)$},\\
(1/\lambda-1)\mu^{|E_C|} X^k \mod p & \mbox{otherwise}.
\end{array}
$$
Thus \eqref{eq:lemzz'1} is zero modulo $p$ if there is a mixed component $C$ in $(V_{\hat{H}},S_0)$.
Assume now
that all components of $(V_{\hat{H}},S_0)$ are pure. The total number
of edges in the components is $|S_0|$, the total number of copies of $\Upsilon$ in the components
is $|V_H|$, and hence
$$
\eqref{eq:lemzz'1} \equiv (1/\lambda - 1)^{\kappa(S_0)} \mu^{|S_0|} X^{|V_H|} \mod p.
$$
\end{proof}

Now we use Claim~\ref{clm:lemmod} to prove \eqref{j10}. Note that by Lemma~\ref{lem:zr},
\begin{equation}\label{eq:lemzz'2}
R'_2(G;\lambda,\mu) = \lambda^{|V_H|\cdot|U|}\sum_{S_0 \subseteq E_{\hat{H}}} \sum_{S \in \Lambda_{S_0}} \lambda^{-\kappa'(S)}\mu^{|S|},
\end{equation}
where $\kappa'(S)$ is the number of pure connected components of $(V_G,S)$.

Note that by Claim~\ref{clm:lemmod} if $S_0$ contains a mixed component then the inner sum
in~\eqref{eq:lemzz'2} is $0$ modulo~$p$. Thus to evaluate~\eqref{eq:lemzz'2} modulo $p$ it
is enough to sum over $S_0$ which contain only pure components. Each such $S_0$ is obtained
from exactly one $S'\subseteq E_H$ by $2$-stretching. Note $(V_{\hat{H}},S_0)$ has the same
number of connected components as $(V_H,S')$, and $|S_0|=2|S'|$. Thus
\begin{eqnarray*}
R'_2(G;\lambda,\mu) & \equiv & \lambda^{|V_H|\cdot|U|}X^{|V_H|}\sum_{S' \subseteq E_H}
(1/\lambda-1)^{\kappa(S')}\mu^{2|S'|}\\
& \equiv & \lambda^{|V_H|\cdot|U|}X^{|V_H|}Z(H;1/\lambda-1,\mu^2) \mod p.
\end{eqnarray*}
\end{proof}

We use different $\Upsilon$ for different values of $\mu$ in the
reduction. When $\mu \neq -2$ we let $\Upsilon_1$ be the
bipartite graph with bipartition $U=\{u_0,u_1\}$, $W=\{v_i \,|\, 0
\leq i \leq k\}$ and $k+2$ edges: edge $\{u_0,v_0\}$, and an edge
between $u_1$ and each $v_i$, for $0 \leq i \leq k$. The specific
vertex of $\Upsilon_1$ is $u_0$. By elementary counting, we have
\begin{equation}\label{eq:gadget}
\begin{split}
\lambda Z'_p(\Upsilon_1;\lambda,\mu) = (\mu+1)^{k+1}+\mu^2+\lambda^{-1}-1,\\
\lambda Z'_p(\Upsilon_1;\lambda,\mu)+Z'_m(\Upsilon_1;\lambda,\mu) = (\mu+1)((\mu+1)^{k+1}+\lambda^{-1}-1).
\end{split}
\end{equation}

When $\mu=-2$ we let $\Upsilon_2$ be the bipartite graph with
bipartition $U=\{u_0,u_1,u_2\}$, $W=\{v_i \,|\, 0 \leq i \leq
2k\}$ and $4k+2$ edges: $\{u_0,v_0\}$, $\{u_1,v_0\}$, and a
complete bipartite graph between $U\setminus\{u_0\}$ and
$W\setminus\{v_0\}$. The specific vertex of $\Upsilon_2$ is $u_0$.
By elementary counting, we have
\begin{equation}\label{eq:gadget2}
\begin{split}
\lambda Z'_p(\Upsilon_2;\lambda,-2) = \lambda^{-2}+5^{2k}\lambda^{-1}-3+3\cdot5^{2k}+\lambda^{-1},\\
\lambda Z'_p(\Upsilon_2;\lambda,-2)+Z'_m(\Upsilon_2;\lambda,-2) = -\lambda^{-2}-5^{2k}\lambda^{-1}-1+5^{2k}+3\lambda^{-1}.
\end{split}
\end{equation}

Fix rational $\lambda \notin \{0,1\}$ and $\mu \neq 0$. We want to
find sufficiently many primes such that there is some integer $k$
for which~\eqref{eq:gadget} satisfies~\eqref{eq:lemmod} (when $\mu
\neq -2$) or~\eqref{eq:gadget2} satisfies~\eqref{eq:lemmod} (when
$\mu=-2$). We need the following result on the density of primes.

\begin{lemma}[\cite{MR1802718,MR1838084,Stevenhagen}]\label{lem:density}
Let $r,q \in \Q^*$ and $q \neq \pm 1$. The density (inside the set of all primes) of primes $p$ such that
\begin{equation*}
q^k \equiv r \mod p,
\end{equation*}
can be satisfied for some integer $k$ is a positive constant, assuming GRH.
\end{lemma}

Lemma~\ref{lem:density} immediately yields the following two corollaries.

\begin{corollary}\label{cor:condition}
Fix rational $\lambda \not\in \{0,1\}$ and rational $\mu \not\in
\{0,-2\}$. The density (inside the set of all primes) of the
primes $p$ such that there is an integer $k$ for
which~\eqref{eq:gadget} satisfies~\eqref{eq:lemmod} is a positive
constant, assuming GRH.
\end{corollary}

\begin{proof}
If $\mu=-1$, then to make~\eqref{eq:gadget} satisfy~\eqref{eq:lemmod}, it is sufficient to have
$\lambda^{-1} \not\equiv 0 \mod p$. Thus, for all but a constant numbers of primes,
and for all positive integers $k$, \eqref{eq:gadget} satisfies~\eqref{eq:lemmod}.

Now assume $\mu \neq -1$. To make~\eqref{eq:gadget} satisfy~\eqref{eq:lemmod}, it is sufficient to have
\begin{equation}\label{eq:condition}
\begin{split}
(\mu+1)^{k+1} +\lambda^{-1}-1 \equiv 0 \mod p,\\
\mu \not\equiv 0 \mod p.
\end{split}
\end{equation}
The corollary follows from Lemma~\ref{lem:density} (and the fact that the number of
primes such that $\mu \equiv 0 \mod p$ is finite).
\end{proof}

\begin{corollary}\label{cor:condition-2}
Fix rational $\lambda \not\in \{0,1\}$ and rational $\mu =-2$. The
density (inside the set of all primes) of the primes $p$ such that
there is an integer $k$ for which~\eqref{eq:gadget2}
satisfies~\eqref{eq:lemmod} is a positive constant, assuming GRH.
\end{corollary}

\begin{proof}
We claim that if
\begin{equation}\label{eq:condition-2}
\begin{split}
(\lambda^{-2}-3\lambda^{-1}+1)(1-\lambda^{-1})^{-1} \equiv 25^k \mod p,\\
1-\lambda^{-1} \not\equiv 25^k \mod p
\end{split}
\end{equation}
then \eqref{eq:gadget2} satisfies~\eqref{eq:lemmod}. The first equation in \eqref{eq:condition-2}
makes the second equation in~\eqref{eq:lemmod} satisfied; and the first and the second equation in
\eqref{eq:condition-2} make the first equation in \eqref{eq:lemmod} satisfied.

By Lemma~\ref{lem:density} the density of primes that make the first equation in~\eqref{eq:condition-2}
satisfied is positive. Solving
$$(\lambda^{-2}-3\lambda^{-1}+1)(1-\lambda^{-1})^{-1}\equiv
1-\lambda^{-1}\mod p,
$$
we obtain $\lambda^{-1}\equiv 0\mod p$ and hence the second equation in~\eqref{eq:condition-2} is
automatically satisfied.
\end{proof}

\begin{proof}[Proof of Lemma~\ref{lem:reduction}]
Let $x,y,\lambda,\mu$ be rational numbers such that~\eqref{j20} is satisfied. Suppose $\lambda = a/b$ and
$\mu=c/d$ with $a,b,c,d \in \Z^*$, $\mathrm{gcd}(a,b)=1$, and $\mathrm{gcd}(c,d)=1$.

Suppose we want to evaluate the Tutte polynomial for $H=(V_H,E_H)$ at $x,y$.
Let $n:=|V_H|$ and $m:=|E_H|$. By~\eqref{eq:tutte_rc}, to evaluate
$T(H;x,y)$, we can instead evaluate $Z(H;1/\lambda-1,\mu^2)$
(note that $\lambda\neq 1$ implies $x\neq 1$ and $y\neq 1$ and hence \eqref{eq:tutte_rc} applies).
Recall that
\begin{equation}\label{eq:rationalrc}
Z(H;1/\lambda-1,\mu^2) = \sum_{S \subseteq E_H} (1/\lambda-1)^{\kappa(S)}\mu^{2|S|} = \frac{L}{a^{n}d^{2m}},
\end{equation}
where
$$L=\sum_{S \subseteq E_H}
(b-a)^{\kappa(S)}a^{n-\kappa(S)}c^{2|S|}d^{2m-2|S|}.$$
Note that $L
\in \Z$ and $|L| \leq 2^m |b-a|^n a^n c^{2m} d^{2m}$.

We now prove the case $\mu \neq -2$. For the case $\mu = -2$, the
proof is similar (by using $\Upsilon_2$ and
Corollary~\ref{cor:condition-2}).

We choose $n^3$ primes $p_1,\ldots,p_{n^3}$ such that
\begin{itemize}
\item $a,b,c,d \not\equiv 0 \mod p_i$, and $a+b \not\equiv 0 \mod p_i$, for each $i \in [n^3]$;
\item there is some integer $k$ for which \eqref{eq:condition} is satisfied with $p=p_i$, for each $i \in [n^3]$;
\item $p_i=O(n^4)$ for $i \in [n^3]$; and
\item $\prod_{i=1}^{n^3}p_i > 2^{m+1}|b-a|^n a^n c^{2m} d^{2m}$.
\end{itemize}
By Corollary~\ref{cor:condition}, these primes exist. We can find them in time polynomial in $n$ by exhaustive search.

For each $p_i$, let $0<k_i<p_i$ be an integer for
which~\eqref{eq:gadget} satisfies~\eqref{eq:lemmod} with $p=p_i$
(by Fermat's little Theorem, if $k_i$ is a solution, then
$k_i+t(p_i-1)$ is a solution as well, for every $t \in \Z$).
Again, we can find them in time polynomial in $n$ by exhaustive
search.

We use $\Upsilon_1$ with $k=k_i$ as above. Let $G_i$ be the
stretch-sum of $H$ and $(\Upsilon_1,u_0)$ as above. Note that
$G_i$ has size polynomial in $n$ since $k_i = O(n^4)$. By
Lemma~\ref{lem:zz'} and~\eqref{eq:rationalrc}, we have
\begin{equation*}
W \equiv a^n d^{2m}\lambda^{-2n} X^{-n} R_2'(G_i;\lambda,\mu) \mod p_i,
\end{equation*}
where $X=\lambda Z'_p(\Upsilon_1;\lambda,\mu)$. We can
make a query to the oracle to obtain the rational number
$R_2'(G_i;\lambda,\mu)$ and thus can compute $L \mod p_i$ in
polynomial time for each $i \in [n^3]$. By the Chinese remainder
theorem, we can compute $L$ in time polynomial in $n$ (see,
e.\,g.,~\cite{MR1406794}, p.106).
\end{proof}

\subsection{Reducing \#{\sc PBIS} to the $R'_2$ polynomial (proof of Theorem~\ref{thm:pbis})}\label{sec:thmpbis}

Now we show the connection between \#{\sc PBIS} and the $R_2'$ polynomial;
the proof of Theorem~\ref{thm:pbis} is similar to the proof of the
high-temperature expansion of the Ising model (see,
e.\,g.,~\cite{MR1237164}).

\begin{proof}[Proof of Theorem~\ref{thm:pbis}]
\begin{eqnarray}
\#\mathrm{PBIS}(G;\eta) & = & \sum_{\sigma:U \cup W \rightarrow \{0,1\}}
(1+\eta)^{w(\sigma)}(1-\eta)^{|E|-w(\sigma)}\nonumber\\
& = & \sum_{\sigma:U \cup W \rightarrow \{0,1\}} \prod_{\{u,v\} \in E}
(1+\eta\chi(\sigma(u),\sigma(v))),\label{eq:pfpbis1}
\end{eqnarray}
where
\begin{equation}\label{choo}
\chi(\sigma(u),\sigma(v))=\Big\{\begin{array}{rl}
 1 & \mbox{if $\sigma(u)=\sigma(v)=1$}\\
-1 & \mbox{otherwise}.
\end{array}
\end{equation}

Let
$$\Psi_{S,\sigma_1,\sigma_2}:=\prod_{\{u,v\} \in S}
\chi(\sigma_1(u),\sigma_2(v))\quad\quad\mbox{and}\quad\quad
\Psi_{S,\sigma_1}:=\sum_{\sigma_2: W \rightarrow \{0,1\}}
\Psi_{S,\sigma_1,\sigma_2}.$$
Expanding the product in \eqref{eq:pfpbis1} and changing the order of summation yields
\begin{eqnarray}
\eqref{eq:pfpbis1} & = &\sum_{\sigma:U \cup W \rightarrow \{0,1\}} \sum_{S \subseteq E} \prod_{\{u,v\} \in S} \eta\chi(\sigma(u),\sigma(v))\nonumber\\
& = & \sum_{S \subseteq E} \eta^{|S|} \sum_{\sigma:U \cup W \rightarrow \{0,1\}} \prod_{\{u,v\} \in S} \chi(\sigma(u),\sigma(v))\nonumber\\
& = & \sum_{S \subseteq E} \eta^{|S|} \sum_{\sigma_1:U \rightarrow \{0,1\}} \Psi_{S,\sigma_1}.\label{eq:pfpbis2}
\end{eqnarray}

Let $N_S(v)$ denote the set of neighbors of $v$ in the subgraph $(U \cup W,S)$. Fix $S$ and
$\sigma_1: U \rightarrow \{0,1\}$. We say that a pair $S,\sigma_1$ is good if for every $v \in W$
the number of vertices $u \in N_S(v)$ such that $\sigma_1(u)=1$ is even. A pair which
is not good will be called bad.

\begin{claim}\label{clm:pair}
$\Psi_{S,\sigma_1}=2^{|W|}(-1)^{|S|}$ if the pair $S,\sigma_1$ is good; and $\Psi_{S,\sigma_1}=0$ if the pair $S,\sigma_1$
is bad.
\end{claim}

\begin{proof}[Proof of Claim~\ref{clm:pair}]
Suppose that $\sigma_2$ and $\sigma_2'$ differ only in
the value assigned to $v\in W$.
For every $u \in N_S(v)$ we have
$$\chi(\sigma_1(u),\sigma_2(v))=\Big\{\begin{array}{rl}
  \chi(\sigma_1(u),\sigma'_2(v)) & \mbox{if}\ \sigma_1(u)=0 \\
  -\chi(\sigma_1(u),\sigma'_2(v)) & \mbox{if}\ \sigma_1(u)=1.
\end{array}$$
Hence,
$$
\Psi_{S,\sigma_1,\sigma_2}=
\Big\{\begin{array}{rl}
  -\Psi_{S,\sigma_1,\sigma'_2} & \mbox{if}\ |\{u \in N_S(v) \,|\, \sigma_1(u)=1\}|\ \mbox{is odd}, \\
  \Psi_{S,\sigma_1,\sigma'_2} & \mbox{otherwise}.
\end{array}
$$
If the pair $S,\sigma_1$ is bad, then there is a vertex $v \in W$
such that $|\{u \in N_S(v) \,|\, \sigma_1(u)=1\}|$ is odd. We
can partition the $W\rightarrow\{0,1\}$ mappings into pairs
$\sigma_2,\sigma'_2$ that differ only in the label of $v$.
For each pair, we have
$\Psi_{S,\sigma_1,\sigma_2}+\Psi_{S,\sigma_1,\sigma'_2}=0$, and
thus $\Psi_{S,\sigma_1}=0$.

If the pair $S,\sigma_1$ is good, then each
$\Psi_{S,\sigma_1,\sigma_2}$ contributes the same value
$(-1)^{|S|}$. Since there are $2^{|W|}$ mappings from $W$ to
$\{0,1\}$ we have $\Psi_{S,\sigma_1}=2^{|W|}(-1)^{|S|}$.
\end{proof}

\begin{claim}\label{clm:rk}
Fix $S \subseteq E$, the number of $\sigma_1: U \rightarrow \{0,1\}$ such that the pair
$S,\sigma_1$ is good is $2^{|U|-\rk(S)}$, where $\rk(S)$ is the rank (over $\F_2$) of the bipartite
adjacent matrix of $(U \cup W,S)$.
\end{claim}

\begin{proof}[Proof of Claim~\ref{clm:rk}]
Let $B$ be the bipartite adjacency matrix of $(U \cup W,S)$. Note that the pair $S,\sigma_1$ is good if and only
if $$\sigma_1^{\mathrm{T}}B \equiv 0 \mod 2$$
(we view $\sigma_1$ as a vector with $(\sigma_1)_v = \sigma_1(v)$).
The claim follows from the fact that the number of vectors $\alpha \in \{0,1\}^{|U|}$
such that $\alpha^{\mathrm{T}}B \equiv 0 \mod 2$ is
$2^{|U|-\rk(B)}$.
\end{proof}

By Claim~\ref{clm:pair} and Claim~\ref{clm:rk}, \eqref{eq:pfpbis2} equals
\begin{equation*}
\sum_{S \subseteq E} 2^{|U|+|W|-\rk(S)}(-\eta)^{|S|} = 2^{|U|+|W|}R'_2(G;1/2,-\eta).
\end{equation*}
\end{proof}

\subsection{Reducing \#{\sc BIS} to \#{\sc PBIS} (proof of Lemma~\ref{lem:reduction2})}\label{sec:lemreduc2}

We use the following construction in the proof of
Lemma~\ref{lem:reduction2}. Given a graph $G=(V,E)$, a prime $p$
and a positive integer $k$ we construct a bipartite graph $G'$ by
replacing each vertex $v\in V$ by a cloud of $kp$ new vertices and
each edge $e\in E$ by a cloud of $p-1$ new vertices. For every
vertex $v\in V$ and every edge $e$ adjacent to $v$ we add a
complete bipartite graph between the cloud of $v$ and the cloud of
$e$.

\begin{lemma}\label{lem:bis2pbis}
Let $p>2$ be a prime, $k$ be a positive integer, and $\eta$ be a rational number
such that $\eta\in\Z_p^*$ and
\begin{equation}\label{eq:condition2}
\begin{split}
\left(\frac{1+\eta}{1-\eta}\right)^{2k} +1 \equiv 0 \mod p,\\
\eta-1 \not\equiv 0 \mod p.
\end{split}
\end{equation}
Let $G=(V,E)$ be a graph and let $G'=(V',E')$ be the graph
constructed above (using $p,k$). Then the number of independent sets of $G$
is congruent to $\#\mathrm{PBIS}(G';\eta) \mod p$.
\end{lemma}

\begin{proof}
We are going to evaluate
\begin{equation}\label{aaa}
\#\mathrm{PBIS}(G';\eta)=\sum_{\sigma:V'\rightarrow\{0,1\}} \prod_{\{u,v\}\in E'}
(1+\eta\chi(\sigma(u),\sigma(v)))
\end{equation}
modulo $p$ (where $\chi$ is given by~\eqref{choo}).

Let $s$ be a vertex of $G$. Consider a set of configurations (that is, $\sigma$'s in
\eqref{aaa}) that differ only in the labels assigned to the vertices in the cloud of $s$ (thus
the configurations agree in the labels assigned to the vertices outside of the cloud of $s$).
Note that by symmetry the value of
$\prod_{\{u,v\}\in E'}
(1+\eta\chi(\sigma(u),\sigma(v)))$ depends only on $h$, the number of vertices in the cloud of $s$
that are assigned $1$. There are ${kp\choose h}$ such configurations and hence the total
contribution of configurations with $h\in\{1,\dots,kp-1\}$ to \eqref{aaa} is zero modulo $p$.
Thus, we only need to consider the $\sigma$'s such that for every $s\in V$ the vertices
in the cloud of $s$ are assigned the same label.

For $\varsigma:V\rightarrow\{0,1\}$
let $C_\varsigma$ be the set of assignments $V'\rightarrow\{0,1\}$ that assign
label $\varsigma(s)$ to every vertex in the cloud of $s$ (for every $s\in V$).
Now we evaluate
\begin{equation}\label{bbb}
\sum_{\sigma\in C_{\varsigma}} \prod_{\{u,v\}\in E'}
(1+\eta\chi(\sigma(u),\sigma(v))).
\end{equation}
Note that each vertex in the edge clouds influences its own set of terms
in the product in~\eqref{bbb}. Thus the sum in~\eqref{bbb} turns into
the following product
\begin{equation}\label{aaaa}
\eqref{bbb} =
\prod_{\{s,t\}\in E}
\Psi(\varsigma(s),\varsigma(t)),
\end{equation}
where
\begin{equation}\label{aaaa2}
\Psi(x,y)=
\left[(1+\eta\chi(x,0))^{kp}(1+\eta\chi(y,0))^{kp}+(1+\eta\chi(x,1))^{kp}(1+\eta\chi(y,1))^{kp}\right]^{p-1}.
\end{equation}
Now we evaluate $\Psi(x,y)$ for $x,y\in\{0,1\}$
\begin{itemize}
\item Case 1: $x=y=1$. We have
\begin{equation*}
\Psi(1,1) = \left[(1-\eta)^{2kp}+(1+\eta)^{2kp}\right]^{p-1},
\end{equation*}
and from the assumption~\eqref{eq:condition2} we obtain
\begin{equation}\label{hh1}
\Psi(1,1) = (1-\eta)^{2kp(p-1)}\left[1+\left(\frac{1+\eta}{1-\eta}\right)^{2kp}\right]^{p-1} \equiv 0 \mod p.
\end{equation}

\item Case 2: $x=y=0$. We have
\begin{equation}\label{hh2}
\Psi(0,0) = \left[2(1-\eta)^{2kp}\right]^{p-1} \equiv 1 \mod p,
\end{equation}
where in the last congruence we used $\eta \not\equiv 1 \mod p$ and Fermat's little Theorem.

\item Case 3: $x=1,y=0$ (the case $x=0,y=1$ is the same). We have
\begin{eqnarray}
\Psi(1,0) & = & \left[(1-\eta)^{2kp}+(1+\eta)^{kp}(1-\eta)^{kp}\right]^{p-1}\nonumber\\
& = & (1-\eta)^{2kp(p-1)}\left[1+\left(\frac{1+\eta}{1-\eta}\right)^{kp}\right]^{p-1}\nonumber\\
& \equiv & 1 \mod p,\label{hh3}
\end{eqnarray}
where the last congruence follows from Fermat's little Theorem and the fact that $((1+\eta)/(1-\eta))^k \not\equiv -1 \mod p$ (since otherwise $((1+\eta)/(1-\eta))^{2k} \equiv 1 \mod p$ which violates \eqref{eq:condition2}).
\end{itemize}

From \eqref{hh1}, \eqref{hh2}, and \eqref{hh3} we obtain that \eqref{aaaa} modulo $p$ is $1$ if
$\varsigma$ corresponds to an independent set of $G$ and is zero otherwise. Thus we have
$$
\#\mathrm{PBIS}(G';\eta) \equiv \sum_{\varsigma\in V\rightarrow\{0,1\}}
\prod_{\{s,t\}\in E} \Psi(\varsigma(s),\varsigma(t)) \equiv \#\mathrm{BIS}(G)\mod p.
$$
\end{proof}

Fix rational $\eta \not\in \{\pm 1, 0\}$. We want to find sufficiently
many primes $p$ such that \eqref{eq:condition2} can be satisfied for some integer $k$.
The following results shows that the density (inside the set of all primes) of such primes is a positive constant.

\begin{lemma}[\cite{MR1257079}]\label{lem:density2}
For $q \in \Q^*$ and $q \neq \pm 1$, the density (inside the set of all primes) of primes for which
there exists $k\in\Z$ such that
\begin{equation*}
q^k \equiv -1 \mod p,
\end{equation*}
is a positive constant.
\end{lemma}

The mapping $\eta\mapsto (\frac{1+\eta}{1-\eta})^2$ maps $\eta\in\Q^*\setminus\{\pm 1\}$
to $\Q^*\setminus\{\pm 1\}$ and hence we have the following.

\begin{corollary}\label{cor:condition2}
Fix rational $\eta \not\in \{\pm 1, 0\}$. The density (inside the
set of all primes) of the primes $p$ such that
~\eqref{eq:condition2} has an integer solution to $k$ is a
positive constant.
\end{corollary}

\begin{proof}[Proof of Lemma~\ref{lem:reduction2}]
The proof is a routine application of the Chinese Remainder
Theorem (similar to the proof of Lemma~\ref{lem:reduction}).
Given an instance $G=(V,E)$ of \#{\sc BIS}, let $n:=|V|$ and $L$
be the number of independent sets of $G$. Note that $L \leq 2^n$.

We choose $n^2$ primes $p_1,\ldots,p_{n^2} > 2$ such that
\begin{itemize}
\item for each $i \in [n^2]$, \eqref{eq:condition2} has an integer solution to
$k$ with $p=p_i$;
\item  $p_i=O(n^3)$ for $i \in [n^2]$, and
\item $\prod_{i=1}^{n^2}p_i > 2^n$.
\end{itemize}
By Corollary~\ref{cor:condition2},
these primes exist when $\eta \not\in \{\pm 1, 0\}$, and we can
find them in time polynomial in $n$ using exhaustive search.

For each $p_i$, let $0<k_i<p_i$ be an integer solution
of~\eqref{eq:condition2} with $p=p_i$ (if $k_i$ is a solution
of~\eqref{eq:condition2} then $k_i+t(p_i-1)$ is a solution for
every $t \in \Z$), again we can find the $k_i$ using exhaustive
search. We then construct a bipartite graph $G'_i$ as above
from $G$, $p_i$ and $k_i$. Note that $G'_i$ has size polynomial in
$n$ since $k_i<p_i = O(n^3)$. As~\eqref{eq:condition2} is
satisfied, by Lemma~\ref{lem:bis2pbis}, we have
\begin{equation*}
L \equiv \#\mathrm{PBIS}(G'_i;\eta) \mod p_i.
\end{equation*}
We can make a query to the oracle to obtain the rational number $\#\mathrm{PBIS}(G'_i;\eta)$ and thus can compute $L \mod p_i$ in polynomial time for each $i \in [n^2]$. By the Chinese remainder theorem, we can compute $L$ in time polynomial in $n$.
\end{proof}

\section{The linear-width of a graph}\label{sec:compact}

For the proof of Theorem~\ref{thm:mix_mm} we will use the
linear-width of a graph, a concept which was first defined by
Thomas~\cite{Thomas96tree-decompositionsof}. In this section
we prove a bound on linear-width in terms of tree-width.

The {\em linear-width} of a graph $G=(V,E)$ is the smallest integer
$\ell$ such that the edges of $G$ can be arranged in a linear
order $e_1,\dots,e_m$ in such a way that, for every $i\in [m]$,
there are at most $\ell$ vertices that have an adjacent edge in
$\{e_1,\dots,e_{i-1}\}$ and an adjacent edge in
$\{e_{i},\dots,e_m\}$. It is known that computing the linear-width
of a graph is NP-complete~\cite{MR1780474}. For paths and cycles
the linear-width is easy to compute.

\begin{example}
The linear-width of a path is $1$. The linear-width of a cycle is $2$.
\end{example}

Let $e_1,\dots,e_m$ be a permutation of the edges of $G=(V,E)$. We say
that a vertex $v \in V$ is {\em dangerous} w.r.t. $i\in [m]$, if there
exist two edges $e_{j},e_{k}$ adjacent to $v$ such that $j < i
\leq k$. Let $D_i$ be the set of vertices which are dangerous
w.r.t. $i$. Note that the linear-width of $G$ is the minimum value
of $\max_i |D_i|$ optimized over all permutations of the edges.

Now we give an upper bound on the linear-width for trees.

\begin{lemma}\label{lem:tree_compact}
Let $T=(V,E)$ be a tree on $n$ vertices. The linear-width of $T$ is at most $\lfloor \log_2 n \rfloor$.
\end{lemma}

\begin{proof}
We describe an ordering $\sigma$ witnessing that the linear-width is at
most $\lfloor \log_2 n \rfloor$. Pick an arbitrary vertex $r \in V$ to be the root.
For a vertex $v$, let $T_v$ be the subtree rooted at $v$ and let $|T_v|$ be the
number of vertices in $T_v$. We do a depth-first search (DFS) on $T$ where the children
are explored in the increasing order of $|T_w|$ (that is, the smaller subtrees
are explored first). Let $\sigma$ be the order in which the edges are discovered
by the DFS.

Next, we show that the linear-width of $\sigma=e_1,\dots,e_m$ is
at most $\lfloor \log_2 n \rfloor$. Let $i\in [m]$. Let
$v_1,\dots,v_j$ be the vertices on the stack right before the edge
$e_i$ was discovered (thus $e_i$ is adjacent to $v_j$). Note that
the vertices which are not on the stack are not dangerous w.r.t. $i$ (they
either are already fully explored or haven't been discovered yet).
Note that if $v_{k+1}$ is the rightmost child of $v_k$ then all
the edges adjacent to $v_k$ are discovered and hence $v_k$
is not dangerous w.r.t. $i$ (for all $k\in [j-1]$). On the other
hand if $v_{k+1}$ is not the rightmost child of $v_k$ then our
exploration order implies
\begin{equation}\label{eeee}
|T_{v_{k+1}}| < |T_{v_k}|/2.
\end{equation}
Note the following two facts: $v_j$ may be dangerous and $|T_{v_j}|\geq 2$.
If there are $\ell$ dangerous vertices among $v_1,\dots,v_{j-1}$ then by
\eqref{eeee} we have $|T_w| \geq 2\cdot 2^{\ell}$ where $w$ is the topmost
dangerous vertex. Using $|T_w|\leq n$ we obtain the result.
\end{proof}

For general graphs we will show a generalization of Lemma~\ref{lem:tree_compact}: a
bound on the linear-width of $G$ in terms of the tree-width of $G$. We now define tree-width
(see, e.\,g.,~\cite{1051910} for a nice treatment).

Given a graph $G=(V,E)$, a {\em tree decomposition} of $G$ is a pair
 $(T,\{U_h\}_{h \in V_T})$ where $T=(V_T,E_T)$ is a tree and $U_h \subseteq V$ satisfy:
(i) each edge of $G$ is in at least one subgraph induced by $U_h$; and (ii)
for any three vertices $t_1,t_2,t_3$ of $T$ such that $t_2$ is in the path between
$t_1$ and $t_3$ in $T$ we have $U_{t_1} \cap U_{t_3}\subseteq U_{t_2}$.
The {\em width} of a decomposition is $\max_{h\in V_T} |U_h|-1$. The {\em tree-width}
of $G$ (denoted ${\rm tw}(G)$) is the minimum width optimized over all tree decompositions.

\begin{lemma}\label{lem:general_compact}
Let $G=(V,E)$ be a graph. Then
$$\mbox{\rm linear-width}(G)\leq ({\rm tw}(G)+1)(\lfloor \log_2 n
\rfloor +1).$$
\end{lemma}

\begin{proof}
Given a graph $G$ with $|V|=n$, let $(T,\{U_h\}_{h \in V_T})$
(where $T=(V_T,E_T)$) be an optimal tree decomposition of~$G$. Let $r:=|V_T|$. We can
assume that $r\leq n$ (see, e.\,g.,~\cite{1051910} p. 579). Let
$\sigma'$  be the order in which the vertices of $T$
are discovered by the DFS of Lemma~\ref{lem:tree_compact}. W.l.o.g.
let $\sigma'=1,\dots,r$.

We will need the following ``separation'' property of dangerous vertices: a path
between an explored and an unexplored vertex must contain a dangerous
vertex (this follows from the fact that an explored vertex with
an unexplored neighbor is dangerous).

Now we construct the ordering $\sigma=e_1,\dots,e_m$ on the edges
of $G$. For each $h \in V_T$, let $G_h$ be the induced subgraph
of $G$ on $U_h$. An edge $e\in E$ may appear in several $G_h$;
we ``assign'' it to the first $G_h$ (in the order $\sigma'$).
Let $\sigma$ be an ordering
such that edges assigned to $G_h$ come before the edges assigned
to $G_{h'}$ for any $h<h'$. Note that we do not impose any
particular order on the edges assigned to $G_h$ (for any $h\in
[r]$).

Let $i\in [m]$ and let $h$ be such that $e_i$ is assigned to $G_h$.
Let $D_i\subseteq V$ be the set of dangerous vertices w.r.t. $i$.
Let $D'_h\subseteq V_T$ be the set of dangerous vertices just
after $h$ was discovered. Let $v\in D_i\setminus U_h$. We will show
that $v\in U_g$ for some $g\in D'_h$.

Let $e_{i_1}$ and $e_{i_2}$ be adjacent to $v$ and such that
$i_1 < i \leq i_2$. Note that $e_{i_1}$ and $e_{i_2}$ cannot
be assigned to $G_h$ (since $v$ is not in $U_h$). Thus
$e_{i_1}$ is assigned to an explored vertex $j_1$ of $T$ and
$e_{i_2}$ is assigned to an unexplored vertex $j_2$ of $T$.
Note that $v$ is in $U_{j_1}$ and $U_{j_2}$ and hence
it is in $U_k$ for all $k$ that are on the path between $j_1$ and
$j_2$ in $T$. By the ``separation'' property of dangerous
vertices one of them (say $g$) must be dangerous, and by the
second property of tree decomposition $U_g$ contains $v$.

Hence,
$$|D_i| \leq |U_h| + \sum_{g\in D'_h} |U_g| \leq ({\rm tw}(G)+1)(\lfloor \log_2 n\rfloor+1).$$
\end{proof}

\section{Analysis of the single bond flip chain for trees}\label{sec:mix_mm}

Given a tree $G=(V,E)$, let $\Omega$ be the set of $2^{|E|}$
subsets of $E$. By Lemma~\ref{l2}, for every $H \subseteq E$,
we know that $\rk(H)$ is the size of maximum matching of
the subgraph $(V,H)$. Let $w(H)$ be the size of maximum
matching in a graph $(V,H)$. Let $P$ be the transition matrix
of the single bond flip Markov chain ${\cal M}$ from definition~\ref{d1}.
It's easy to see that ${\cal M}$ is ergodic with unique stationary distribution
$\pi$ such that $$\pi(H) \propto \lambda^{w(H)}\mu^{|E|}.$$

The goal of this section is to prove Theorem~\ref{thm:mix_mm}.

\subsection{The canonical paths}

We will bound the mixing time of our chain ${\cal M}$ using the
canonical paths method, introduced
in~\cite{MR1097463,MR1211324,MR1025467}. Now we go over the basic
definitions for Markov chains, see, e.\,g., \cite{MR1960003} for a
comprehensive background.

\begin{definition}
The {\em total variation distance} of two probability
distribution $\nu$ and $\nu'$ on $\Omega$ is
$$\|\nu-\nu'\|_{TV}=\frac{1}{2}\sum_{H \in \Omega}|\nu(H)-\nu'(H)|=\max_{{\cal S} \subseteq \Omega}|\nu({\cal S})-\nu'({\cal S})|.$$
\end{definition}

\begin{definition}
The {\em mixing time} from initial
state $H$, $\tau_H(\varepsilon)$, is defined as
$$\tau_H(\varepsilon)=\min\{t:\|P^t(H,\cdot)-\pi\|_{TV} \leq \varepsilon\},$$
and the \emph{mixing time} $\tau(\varepsilon)$ of the chain is
defined as
$$\tau(\varepsilon)=\max_{H \in \Omega}\{\tau_H(\varepsilon)\}.$$
\end{definition}

Let $\sigma=e_1,\dots,e_m$ be an ordering of the edges of $G=(V,E)$
(we will usually use the orderings supplied by Lemma~\ref{lem:tree_compact} or
Lemma~\ref{lem:general_compact}). Given any pair $I,F
\in \Omega$, let $I \oplus F$ be the symmetric difference of $I$
and $F$ (that is, the set of edges which are in either $I$ or $F$
but not in both). We define a canonical path $\gamma_{I,F}$
between $I$ and $F$ as follows. Let $e_{i_1},\dots,e_{i_k}$ be the edges
from $I\oplus F$ ordered according to $\sigma$ (that is, $i_1<i_2<\dots<i_k$).
Let
\begin{equation}\label{canp}
\gamma_{I,F}=(H_0,H_1,\ldots,H_{k}),
\end{equation}
where $H_0=I$, $H_k=F$ and $H_{j} = H_{j-1} \oplus \{e_{i_j}\}$.

\begin{lemma}\label{lem:dif_mm}
Let $G=(V,E)$ be a graph. Let $\sigma=e_1,\dots,e_m$ be an ordering on $E$ with linear-width $\ell$.
Let $I,F$ be subsets of $E$ and let $H$ be on the canonical path~\eqref{canp} (that is, $H=H_j$ for some $j\in\{0,\dots,k\}$).
Then
$$|w(I)+w(F)-w(H)-w(C)| \leq \ell,$$
where $C=I\oplus F\oplus H$.
\end{lemma}

\begin{proof}
Let $Q=\{e_1,\dots,e_{i_j}\}$. Note that $H = (F\cap Q)\cup (I\cap Q^c)$,
where $Q^c$ is the complement of $Q$ (that is, $E\setminus Q$).
Similarly, $C=(I\cap Q)\cup (F\cap Q^c)$.

Let $D$ be the set of dangerous vertices w.r.t. $e_{i_j+1}$. Let $M_I$ and $M_F$ be
maximum matchings of $I$ and $F$, respectively. Let
$$
M_H = (M_F\cap Q)\cup (M_I\cap Q^c)\quad\mbox{and}\quad M_C = (M_I\cap Q)\cup (M_F\cap Q^c).
$$
Note that all vertices of $M_H$ with degree $\geq 2$ are in $D$ (a vertex which is not $D$
has all adjacent edges (in $G$) from $Q$ or from $Q^c$ and hence the adjacent edges (in $M_H$)
agree with $M_I$ or $M_F$). The same is true for $M_C$. Moreover if a vertex $v\in D$ has
degree $2$ in $M_H$ then it has degree $0$ in $M_C$. Thus by removing $\leq |D|$ edges from
$M_H$ and $M_C$ we can turn both of them into matchings. Thus
\begin{equation}\label{yy1}
w(H)+w(C)\geq w(I)+w(F)-|D|\geq w(I)+w(F)-\ell.
\end{equation}
Note that a canonical path from $I':=H$ to $F':=C$ passes through $H':=I$ (with
$C':=I'\oplus F'\oplus H'=F$). Thus
\begin{equation}\label{yy2}
w(I)+w(F) = w(H')+w(C')\geq w(I')+w(F')-\ell = w(H) + w(C) - \ell.
\end{equation}
Combining \eqref{yy1} and \eqref{yy2} we get the lemma.
\end{proof}

\subsection{The congestion of ${\cal M}$}

Now we analyze the congestion of the collection $\Gamma=\{\gamma_{I,F}\,|\,I,F\in\Omega\}$
where $\gamma_{I,F}$ are canonical paths defined in~\eqref{canp}.
For each transition $(H,H')$ such that $P(H,H')>0$, let $cp(H,H')$ be the set of pairs $(I,F)$ such that
$(H,H')\in \gamma_{I,F}$. The congestion of $\Gamma$ on $(H,H')$
is (see, e.\,g., \cite{MR1960003})
\begin{equation}\label{eq:cong}
\varrho_{(H,H')}=\frac{1}{P(H,H')}\sum_{I,F: (H,H') \in \gamma_{I,F}} \frac{\pi(I)\pi(F)}{\pi(H)}|\gamma_{I,F}|,
\end{equation}
where $|\gamma_{I,F}|$ is the length of $\gamma_{I,F}$. The congestion of $\Gamma$
is defined as
\begin{equation*}
\varrho := \max_{(H,H'): \atop P(H,H')>0} \varrho_{(H,H')}.
\end{equation*}

We will use the following connection between the congestion and the mixing time.

\begin{theorem}[\cite{MR1097463,MR1211324}]\label{thm:mix_cong}
$\tau_H(\varepsilon) \leq \varrho(\log(1/\pi(H))+\log(1/\varepsilon))$ for each starting state $H \in \Omega$.
\end{theorem}

At the end of this section we prove the following bound on the congestion of $\Gamma$.

\begin{lemma}\label{lem:cong_mm}
Let $G=(V,E)$ be a graph.
Let $\sigma=e_1,\dots,e_m$ be an ordering on $E$ with
linear-width~$\ell$. For every $(H,H')$ such that $P(H,H')>0$, and
for every $\lambda,\mu > 0$ we have
$$\varrho_{(H,H')} \leq
2|E|^2\bar{\lambda}^{\ell},$$
where $\bar{\lambda}=\max\{\lambda,1/\lambda\}$.
\end{lemma}

We can now prove Theorem~\ref{thm:mix_mm}.

\begin{proof}[Proof of Theorem~\ref{thm:mix_mm}]
Since $G=(V,E)$ is a tree, by Lemma~\ref{lem:tree_compact},
we have $\ell\leq\lfloor\log_2 n \rfloor$,
by Lemma~\ref{lem:cong_mm}, we have
$$\varrho\leq 2|E|^2\bar{\lambda}^{\ell}\leq 2|E|^2 n^{|\log_2\lambda|}\leq 2 n^{2+|\log_2\lambda|}.$$
Theorem~\ref{thm:mix_mm} now follows from Theorem~\ref{thm:mix_cong}.
\end{proof}

Now we bound the congestion of our canonical paths.

\begin{proof}[Proof of Lemma~\ref{lem:cong_mm}]
We will bound $\varrho_{(H,H')}$ for every $(H,H')$ such that
$P(H,H')>0$. Let $\hat{H}=H$ if $\pi(H) \leq \pi(H')$ and
$\hat{H}=H'$ otherwise. Note that
\begin{equation}\label{eq:encode}
\frac{\pi(\hat{H})}{2|E|}= \pi(H)P(H,H') = \pi(H')P(H',H),
\end{equation}
since ${\cal M}$ is reversible. We define a mapping $f: cp(H,H')
\rightarrow \Omega$ such that $f(I,F)=I \oplus F \oplus \hat{H}$ for
every pair $(I,F) \in cp(H,H')$.

First, note that $f$ is an injection. Given $J \in \Omega$
we can determine the unique $I,F$ such that $f(I,F)=J$, by first
computing $J \oplus \hat{H}$, and the using the ordering
$\sigma$ on the edges of $G$ to recover $I$ and $F$.

Note that
\begin{equation}\label{eq:edge}
|I|+|F| = |\hat{H}|+|f(I,F)|,
\end{equation}
and
\begin{equation}\label{eq:mm}
|w(I)+w(F)-w(\hat{H})-w(f(I,F))| \leq \ell,
\end{equation}
where~\eqref{eq:mm} follows from Lemma~\ref{lem:dif_mm}.

Let $L=\sum_{J} \lambda^{w(J)}\mu^{|J|}$. We have the following
upper bound on $\varrho_{(H,H')}$. By~\eqref{eq:cong} and~\eqref{eq:encode}, we have
\begin{eqnarray}
\varrho_{(H,H')} & = & 2|E|\sum_{(I,F)\in cp(H,H')}\frac{\pi(I)\pi(F)}{\pi(\hat{H})}|\gamma_{I,F}|\nonumber\\
& = & 2|E|^2\sum_{(I,F)\in cp(H,H')} \frac{\lambda^{w(I)+w(F)-w(\hat{H})}\mu^{|I|+|F|-|\hat{H}|}}{L}\nonumber\\
& \leq & 2|E|^2\bar{\lambda}^{\ell}\sum_{(I,F)\in cp(H,H')} \frac{\lambda^{w(f(I,F))}\mu^{|f(I,F)|}}{L}\label{eq:xxxx}\\
& \leq & 2|E|^2\bar{\lambda}^{\ell}\label{eq:yyyy},
\end{eqnarray}
where \eqref{eq:xxxx} follows from~\eqref{eq:edge} and~\eqref{eq:mm}, and~\eqref{eq:yyyy} follows from the fact that $f$ is an injection from $cp(H,H')$ to $\Omega$.
\end{proof}

\section{The single bond flip chain for the random cluster
model}\label{S:rc}

The following sampling problem corresponds to the random cluster model defined in \eqref{eq:rc}.

\vskip 0.2cm
\noindent\textsc{Random Cluster Model} with $q,\mu \geq 0$ ({\sc RC}$(q,\mu)$)

Instance: a graph $G=(V,E)$.

Output: $S \subseteq E$ with probability of $S \propto q^{\kappa(S)}\mu^{|S|}$.
\vskip 0.2cm

The random cluster model was introduced by Fortuin and
Kasteleyn~\cite{MR0359655} as a model in statistical physics, and
was intensively studied since then (see, e.\,g., the monograph by
Grimmett~\cite{MR2243761}). By the results of~\cite{MR1049758,MR1179248,MR2201454},
the problem of exactly computing the partition function of {\sc
RC}$(q,\mu)$) is \#P-hard even in bipartite planar graphs for
all rational numbers $q,\mu > 0$ except when (i) $q=1$; and (ii) $q=2$.

In the context of approximate counting, little is known about \#{\sc
RC}$(q,\mu)$. The case $q=2$ and $\mu>0$ corresponds to the
ferromagnetic Ising model and has an FPRAS~\cite{MR1237164}.
For ``dense graphs''
an FPRAS exists for all $q,\mu \geq 0$,~\cite{MR1300965,MR1368847}.

Cooper and Frieze~\cite{MR1716764} and Cooper et
al.~\cite{MR1757967} investigated the Swendsen-Wang algorithm for
{\sc RC}$(q,\mu)$. They proved, using coupling and
conductance methods, that the Swendsen-Wang algorithm is rapidly
mixing on some classes of graphs. Gore and Jerrum~\cite{MR1733467}
pointed out that the single bond flip chain for random cluster model is
torpid (mixes in exponential time) on the complete graph $K_n$ for
$q \geq 3$ and a value of $\mu$ depending on the size of the graph
(that is, $\mu$ is not a fixed constant).

We will show that the single bond flip chain mixes for graphs with
bounded tree-width. Our motivation is ``showing that Markov chain mixes'',
not approximation of \#{\sc RC}$(q,\mu)$ (since the
problem of computing the partition function of {\sc RC}$(q,\mu)$
is polynomial-time solvable~\cite{MR1192381} for bounded
tree-width graphs).

The single bond flip for the random cluster model is analogous to the
chain from Definition~\ref{d1}. At time $t$, pick an edge
$e\in E$ at random and let $S=X_t\oplus \{e\}$. Set $X_{t+1}=S$
with probability
$(1/2)\min\{1,q^{\kappa(S)-\kappa(X_t)}\mu^{|S|-|X_t|}\}$
and $X_{t+1}=X_t$ with the remaining probability.
The chain ${\cal M}$ defined above is ergodic with
unique stationary distribution $\pi$ such that
$\pi(H) \propto q^{\kappa(H)}\mu^{|H|}$.
Let $P$ be the transition matrix of ${\cal M}$.

We will bound the mixing time of $\cal M$ for every
fixed constant $q,\mu > 0$. Given an ordering $\sigma$ on the
edges of $G$, we define the same canonical paths for each pair of
subgraphs $I,F$, and the same mapping $f$ for every $(H,H')$ such
that $P(H,H')>0$, as in Section~\ref{sec:mix_mm}. We present a
result analogous to Lemma~\ref{lem:dif_mm}.

\begin{lemma}\label{lem:dif_rc}
Let $G=(V,E)$ be a graph.
Let $\sigma=e_1,\dots,e_m$ be an ordering on $E$ with linear-width $\ell$. Let $I,F$ be subgraphs
of $E$ and let $H$ be on the canonical path~\eqref{canp} (that is, $H=H_j$ for some $j\in\{0,\dots,k\}$).
Then
$$|\kappa(I)+\kappa(F)-\kappa(H)-\kappa(C)| \leq \ell,$$
where $C=I\oplus F\oplus H$.
\end{lemma}

\begin{proof}
As in the proof of Lemma~\ref{lem:dif_mm} let
$Q=\{e_1,\dots,e_{i_j}\}$. Let $D$ be the set of dangerous
vertices w.r.t. $e_{i_j+1}$. We can split $V\setminus D$ into
$V_1\cup V_2$, where $V_1$ contains vertices whose adjacent edges
are from $Q$ and $V_2$ contains vertices whose adjacent edges are
from $Q^c$.

Let $I_i$ be the subgraph of $I$ induced by $V_i\cup D$ and
$F_i$ be the subgraph of $F$ induced by $V_i\cup D$ (for $i=1,2$).
Note that $H$ can be viewed as a gluing of $I_2$ and $F_1$
on $D$; we will denote this $H=I_2\cup F_1$. Similarly
$C=I_1\cup F_2$.

First note that
$$\kappa(F_1) + \kappa(F_2)\leq \kappa(F)+|D|,$$
since $F_1,F_2$ can by obtained from $F$ by
splitting each vertex in $D$ and the splitting
of one vertex increases the number of components
by at most one.

Let $I'$ consist of the connected components of $I$
that do not contain a vertex from $D$. Let $I'_i$ be
the subgraph of $I$ induced by $V_i$ (for $i=1,2$).
Note that
\begin{equation}\label{qqqqq}
\begin{split}
\kappa(F_1\cup I'_2) + \kappa(F_2\cup I'_1)=
\kappa(F_1) + \kappa(I'_2) + \kappa(F_2) + \kappa (I'_1)=\\
\kappa(F_1) + \kappa(F_2) + \kappa (I')\leq \kappa(F)+|D|+\kappa(I')\leq
\kappa(F)+\kappa(I)+|D|,
\end{split}
\end{equation}
where the first and second inequalities follow from the fact that the following
pairs of graphs are vertex disjoint: $F_1$ and $I'_2$, $F_2$ and $I'_1$, and $I'_1$ and $I'_2$.

Note that $I_2\oplus I'_2$ consists of connected components that
contain a vertex from $D$. Thus $\kappa(F_1\cup I_2) =
\kappa(F_1\cup I_2')$ (since adding a component from $I_2\oplus
I'_2$ to $F_1\cup I_2'$ only attaches that component to an
existing component). Similarly $\kappa(F_2\cup I_1) =
\kappa(F_2\cup I_1')$. Combining with \eqref{qqqqq} we obtain
\begin{equation}\label{yy1-2}
\kappa(H)+\kappa(C)\leq \kappa(I)+\kappa(F)+|D|\leq \kappa(I)+\kappa(F)+\ell.
\end{equation}
Note that a canonical path from $I':=H$ to $F':=C$ passes through $H':=I$ (with
$C':=I'\oplus F'\oplus H'=F$). Thus
\begin{equation}\label{yy2-2}
\kappa(I)+\kappa(F) = \kappa(H')+\kappa(C')\leq \kappa(I')+\kappa(F')+\ell = \kappa(H) + \kappa(C) + \ell.
\end{equation}
Combining \eqref{yy1-2} and \eqref{yy2-2} we get the lemma.
\end{proof}

We have the following bound on the congestion.

\begin{lemma}\label{lem:cong_rc}
Let $G=(V,E)$ be a graph. Let $\sigma=e_1,\dots,e_m$ be an ordering on $E$ with
linear-width $\ell$. For every $(H,H')$ such that $P(H,H')>0$, and
for every $q,\mu >0$ we have
$$\varrho_{(H,H')} \leq
2|E|^2\bar{q}^{\ell},$$
where $\bar{q}=\max\{q,1/q\}$.
\end{lemma}

\begin{proof}
Same as the proof of Lemma~\ref{lem:cong_mm}.
\end{proof}

By Theorem~\ref{thm:mix_cong} and  Lemma~\ref{lem:cong_rc},
we have the following result.

\begin{theorem}
Given an graph $G=(V,E)$ with $m=|E|$, $n=|V|$ and linear-width $\ell$, the mixing time $\tau(\varepsilon)$ of
$\cal M$ is
$$\tau(\varepsilon)=O\left(m^2\bar{q}^{\ell}(m+n|\log q|+m|\log\mu|+\log(1/\varepsilon))\right),$$
where $\bar{q}=\max\{q,1/q\}$.
\end{theorem}

By~Lemma~\ref{lem:general_compact} we conclude the following.

\begin{corollary}
Let $k$ be an integer and let $q,\mu>0$. Let $G=(V,E)$ be a graph with tree-width bounded by~$k$, $n:=|V|$, and
$m:=|V|$. The mixing time $\tau(\varepsilon)$ of the single bond flip chain ${\cal M}$ for {\sc RC}($q,\mu$)
is bounded as follows
$$\tau(\varepsilon)=
O\left(m^2\bar{q}^{(k+1)(1+\log_2 n)}(m+n|\log q|+m |\log\mu|+\log(1/\varepsilon))\right),
$$
where $\bar{q}=\max\{q,1/q\}$.
\end{corollary}

We note that we do not need to find the linear-width (or tree-width) of a graph
efficiently, the existence result is enough for our proof of the mixing time.

\section{Conclusions}

We conclude with an observation that a generalization of {\sc RWM}$(\lambda,\mu)$
does not have an FPRAS (unless RP=NP) and a few questions.

Let $A$ be an $m \times n$ matrix whose entries are zeros, ones,
and indeterminates, where each indeterminate occurs once. A {\em
completion} of $A$ is a substitution of $0,1$ to all the
indeterminates in $A$. We denote ${\cal C}_A$ to be the set of all
completions of~$A$. Let $\rk(B)$ be the rank of $B$ over
$\F_2$. Can we sample $B$ from ${\cal C}_A$ with the probability
of $B$ proportional to $\lambda^{\rk(B)}$? Note that
this problem is a generalization of the {\sc RWM}$(\lambda,1)$
problem. It turns out that finding the minimum rank completion of
a matrix is NP-hard (Proposition 2.1,~\cite{MR1417351}) and hence
a sampler is unlikely (unless NP=RP), since for $\lambda=2^{-n^2}$
a random completion will be the minimum rank completion (with
constant probability). The sampling problem could be easy for
sufficiently large $\lambda$ (the problem of finding maximum rank
completion is in P, see, e.\,g., Section 4.1 of
~\cite{MR2298298}).

\begin{question}
What other interesting properties are encoded by the polynomial?
\end{question}

\begin{question}
Can one sample maximum rank completions of a matrix?
\end{question}

\bibliographystyle{plain}
\bibliography{bibfile}

\end{document}